\documentclass{article}

\usepackage{blindtext} 

\usepackage[sc]{mathpazo} 
\usepackage[T1]{fontenc} 
\usepackage{microtype} 
\usepackage{bm} 
\usepackage[english]{babel} 
\usepackage{todonotes}
\usepackage[hmarginratio=1:1,top=32mm,columnsep=20pt]{geometry} 
\usepackage[hang, small,labelfont=bf,up,textfont=it,up]{caption} 
\usepackage{booktabs} 

\usepackage{lettrine} 

\usepackage{enumitem} 
\setlist[itemize]{noitemsep} 

\usepackage{abstract} 

\usepackage{titlesec} 
\renewcommand\thesection{\Roman{section}} 
\renewcommand\thesubsection{\arabic{subsection}} 
\titleformat{\section}[block]{\large\scshape\centering}{\thesection.}{1em}{} 
\titleformat{\subsection}[block]{\large}{\thesubsection.}{1em}{} 

\usepackage{fancyhdr} 
\usepackage{titling} 
\usepackage{graphicx}

\usepackage{hyperref} 
\usepackage[normal,tight,center]{subfigure}
\usepackage{xcolor}
\usepackage{float}

\usepackage{amsmath}
\usepackage{amsthm}
\usepackage{amsfonts}
\newtheorem{thm}{Theorem}
\newtheorem{deft}{Definition}
\newtheorem{lem}{Lemma}

\usepackage{graphicx}

\newcommand\numberthis{\addtocounter{equation}{1}\tag{\theequation}}

\pretitle{\begin{center}\Huge\bfseries} 
\posttitle{\end{center}} 
\title{Expectile Neural Networks for Genetic Data Analysis of Complex Diseases} 
\author{%
\textsc{Jinghang Lin$^{1, 2}$, Xiaoran Tong$^{2,3}$, Chenxi Li$^{3}$, Qing Lu$^{2}$}\\
\normalsize 1.Department of Statistics and Probability, Michigan State University, East Lansing, Michigan 48823, U.S.A.\\
\normalsize 2.Department of Biostatistics, University of Florida, Gainesville, Florida 32611, U.S.A.\\
\normalsize 3.Department of Epidemiology and Biostatistics, Michigan State University, East Lansing, Michigan 48832, U.S.A\\
\normalsize \href{mailto:linjingh@msu.edu}{linjingh@msu.edu,} \href{mailto:lucienq@ufl.edu}{lucienq@ufl.edu}
}
\date{} 
\renewcommand{\maketitlehookd}{%
\begin{abstract}
The genetic etiologies of common diseases are highly complex and heterogeneous. Classic statistical methods, such as linear regression, have successfully identified numerous genetic variants associated with complex diseases. Nonetheless, for most complex diseases, the identified variants only account for a small proportion of heritability. Challenges remain to discover additional variants contributing to complex diseases. Expectile regression is a generalization of linear regression and provides completed information on the conditional distribution of a phenotype of interest. While expectile regression has many nice properties and holds great promise for genetic data analyses (e.g., investigating genetic variants predisposing to a high-risk population), it has been rarely used in genetic research. In this paper, we develop an expectile neural network (ENN) method for genetic data analyses of complex diseases. Similar to expectile regression, ENN provides a comprehensive view of relationships between genetic variants and disease phenotypes and can be used to discover genetic variants predisposing to sub-populations (e.g., high-risk groups). We further integrate the idea of neural networks into ENN, making it capable of capturing non-linear and non-additive genetic effects (e.g., gene-gene interactions). Through simulations, we showed that the proposed method outperformed an existing expectile regression when there exist complex relationships between genetic variants and disease phenotypes. We also applied the proposed method to the genetic data from the Study of Addiction: Genetics and Environment(SAGE), investigating the relationships of candidate genes with smoking quantity.

\textbf{keywords: nonlinear, gene-gene interaction, expectile regression, neural networks}
\end{abstract}
}

\begin{document}

\maketitle


\section{Introduction}
Converging evidence suggests that the genetic etiologies of complex diseases are highly heterogeneous \cite{GHHD}, and various genetic factors and environmental determinants could play different roles in subgroups of the population. Linear regression has been commonly used in genetic studies to investigate the effects of genetic variants on the mean of a continuous phenotype. However, if we are interested in a complete view of genetic effects across the entire distribution of phenotypes or are interested in investigating genetic contribution to a sub-population(e.g., a high-risk population), quantile regression and expectile regression are great alternative choices (\cite{RQ}, \cite{ALSE}). Quantile regression generalizes median regression and has been widely used in fields such as economics \cite{QRBC}, medicine (\cite{CAHT}, \cite{QRML}), and environmental science (\cite{ACSR}) to study entire conditional distributions of responses given covariates.  While quantile regression has many good properties (e.g., being robust to distribution assumption and outlies), as pointed out by Newey and Powell \cite{ALSE}, quantile regression has several limitations. First, quantile regression uses the check function with the absolute least error as loss function, which is not continuously differentiable and is computationally difficult for parameter estimation. Second, quantile regression is relatively inefficient for error distributions that are close to Gaussian or have low densities at the corresponding percentile. Third, it is challenging to estimate the density function values of quantile regression. 

To address these issues, Newey and Powell \cite{ALSE} proposes expectile regression, which uses the sum of asymmetric residual squares as the loss function. Since the loss function is convex and differentiable, expectile regression has a computational advantage over quantile regression. Similar to quantile regression, expectile regression makes no assumption on error distribution (e.g., homoscedasticity) and can be used to study the entire distribution of the responses. Expectile regression can be viewed as a generalization of linear regression. A typical expectile regression assumes a linear relationship between the expectile and the covariates, which may not be suitable for genetic data analysis as genetic variants likely influence phenotypes in a complicated manner (e.g., through interactions) (\cite{DGGI}). Simply considering linear and additive genetic effects can't fully take this complexity into account. 

In this paper, we integrate the idea of neural networks into expectile regression and develop an expectile neural network (ENN) method to model the complex relationship between genotypes and phenotypes. While several methods have been developed to integrate neural networks into quantile regression(\cite{NNRQ} ,\cite{QRNN} ,\cite{AQRNN}), few studies have been focused on investigating nonlinear expectile regressions, especially using neural networks. Compared to quantile regression neural networks(QRNN), ENN has several advantages. The empirical loss function in ENN is differentiable everywhere. Moreover, ENN can detect the heteroscedasticity in the data since ENN is more sensitive to extreme values than QRNN.

The rest of the paper is organized as follows: in Section 2, we review expectile regression and propose an ENN method. We then give an inequality that bounds the integrated squared error of an expectile function estimator in terms of risk functions. The proof of inequality is detailed in the Appendix. Simulations were conducted in Section 3 to evaluate the performance of the new method. In Section 4, we applied ENN to the SAGE data, studying genetic contribution to smoking quantity. We provide the summary and concluding remarks in Section 5. 

\section{Method}
In this section, we briefly introduce expectile regression and then propose an expectile neural network. Suppose we have $n$ samples,$\lbrace (\mathbf{x_i},y_i),i=1,...,n \rbrace$, where $\mathbf{x_i}=(1,x_{i,1},...,x_{i,p})^T$ and $y_i$ denote a $p-$dimensional covarites and the response for the $i$th sample, respectively. In this paper, the covariates are primarily genetic variants, such as single nucleotide polymorphisms (SNPs), which are typically coded as the number of minor frequent allele (e.g., AA=2, Aa=1, aa=0). The covariates $\mathbf{x_i}$ can also include personal characteristics (e.g., gender) and environmental determinants.

\subsection{Expectile regression}
Given the data, linear regression is commonly used to model the relationship between the covariates and the mean response. However, if we want to explore a complete relationship between the covarites and the response (e.g., genetic contribution to a high-risk population), an expectile regression can be used. The expectile regression for the $\tau-$expectile can be expressed as,
\begin{equation}
Expectile(\tau)= \mathbf{x}^T \hat{\bm{\beta}},    
\end{equation}
where $\hat{\bm{\beta}}$ is the estimator of coefficients $\bm{\beta}=(\beta_0, \beta_1,...,\beta_p)^T$.  The regression parameters, $\bm{\beta}$, can be obtained by minimizing an asymmetric $L_2$ loss function,
\begin{equation}
\mathcal{R}_{L_{\tau}}(\beta;\tau)=\frac{1}{n} \sum_{i=1}^{n}L_{\tau}(y_{i},\mathbf{x_i}^T\bm{\beta}),
\end{equation}
where $L_{\tau}(\cdot)$ is asymmetric squared loss with convex form
\begin{equation}
L(y_{i}, \mathbf{x_i}^T\bm{\beta})=\left\{
\begin{aligned}
&(1- \tau) (y_{i} - \mathbf{x_i}^T\bm{\beta})^2, & if \ y_{i} < \mathbf{x_i}^T\bm{\beta} \\
&\tau (y_i - \mathbf{x_i}^T\bm{\beta})^2, & if \ y_i \geq \mathbf{x_i}^T\bm{\beta}.
\end{aligned}
\right.
\end{equation}

For a model with a large $p$, a penalty term can be added to the risk function to reduce the model complexity, 
\begin{equation}
\mathcal{R}_{L_{\tau}}(\beta;\tau)=\frac{1}{n} \sum_{i=1}^{n}L_{\tau}(y_{i}-\mathbf{x_i}^T\bm{\beta}) + \lambda\sum_{i=1}^{p} \beta_{i}^2. 
\end{equation}

When $\tau = 0.5$, the corresponding expectile regression degenerates to a standard linear regression. Therefore, expectile regression can also be viewed as a generalization of linear regression.

\subsection{Expectile neural network}
A typical expectile regression model focuses on linear relationships between covariates and responses. In reality, the underlying relationship could be non-linear and involve complicated interactions among covariates. In order to model complex relationships between covariates and responses, we integrate the idea of neural networks into expectile regression and propose an ENN method. ENN can be considered as a nonparametric expectile regression or neural networks with asymmetric $L_2$ loss function, where we don't assume a particular functional form of covariates and use neural networks to approximate the underlying expectile regression function. We illustrate ENN with one hidden layer. The method can be easily extended to an expectile regression deep neural network with multiple layers.

\begin{figure}[h]
\begin{center}
\includegraphics[scale=0.6]{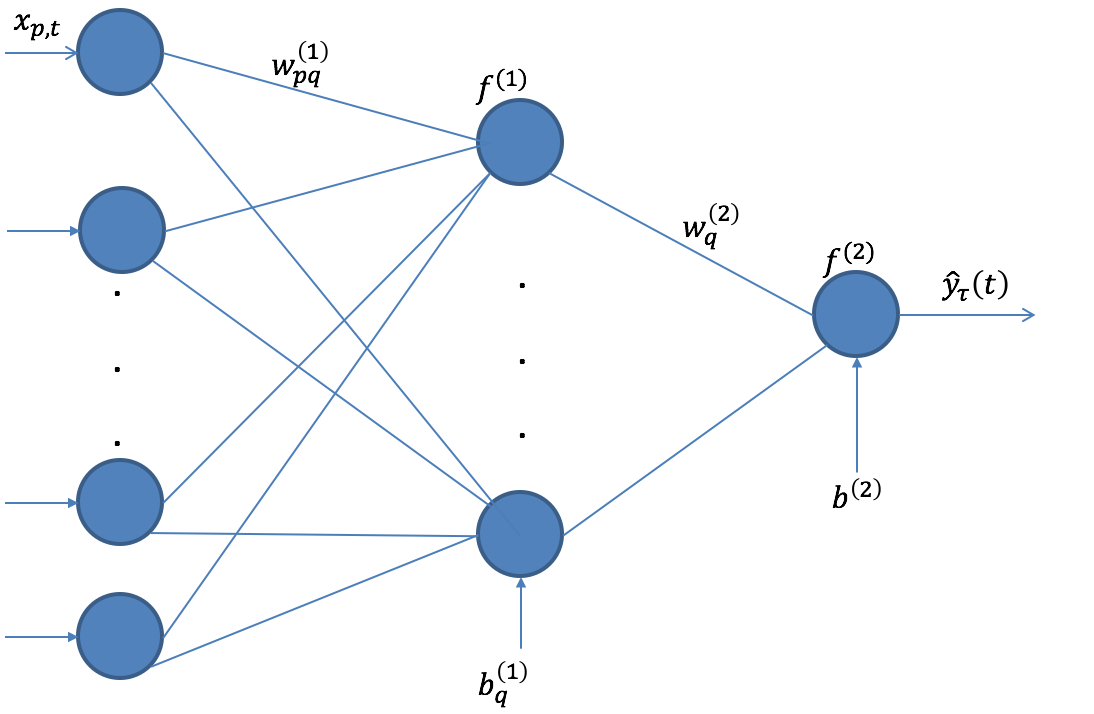}
\caption{A graphical representation of expectile neural network}
\end{center}
\end{figure}

Given the covariates $\mathbf{x}_t$, we first build the hidden nodes $h_{q,t}$,
\begin{equation}\label{hiddenlayer}
h_{q,t}= f^{(1)}(\sum_{p=1}^{P} x_{p,t}w_{pq}^{(1)}+b_{q}^{(1)}),  q= 1,...,Q, t = 1, ..., n.
\end{equation}
where $w_{pq}$ denotes weights and $b_{q}$ denotes the bias;  $f^{(1)}$ is the activation function for the hidden layer that can be a sigmoid function, a hyperbolic tangent function, or a rectified linear units(ReLU) function. Similar to hidden nodes in neural networks, the hidden nodes in ENN can learn complex features from covariates $\mathbf{x}$, which makes ENN capable of modelling non-linear and non-additive effects. Based on these hidden nodes, we can model the conditional $\tau$-expectile, $\hat{y}_{\tau}(t)$,
\begin{equation}\label{output}
\hat{y}_{\tau}(t)=f^{(2)}(\sum_{q=1}^{Q} h_{q,t}w_{q}^{(2)}+b^{(2)}),
\end{equation}
where $f^{(2)}$, $w_{q}^{(2)}$, and $b^{(2)}$ are the activation function, weights, and bias in the output layer, respectively. $f^{(2)}$ can be an identity function, a sigmoid function, or a rectified linear units(ReLU) function.  A graphical representation of ENN is given in Figure 1.

From equations (\ref{hiddenlayer}) and (\ref{output}), we can have the ENN model:
\begin{equation}
\hat{y}_{\tau}(t) = f^{(2)}(\sum_{q=1}^{Q} f^{(1)}(\sum_{p=1}^{P} x_{p,t}w_{pq}^{(1)}+b_{q}^{(1)})w_{q}^{(2)}+b^{(2)}). 
\end{equation}
To estimate $w_{pq}^{(1)}, b^{(1)}_{q}, w_{q}^{(2)}, b^{(2)}$, we minimize the empirical risk function

\begin{equation}
\mathcal{R}(\tau)=\frac{1}{n}\sum_{i=1}^{n} L_{\tau}(y_i,f(\mathbf{x_i})),
\end{equation}
where
\begin{equation}
L_{\tau}(y_{i}, f(\mathbf{x_i}))=\left\{
\begin{aligned}
&(1- \tau) (y_{i} - f(\mathbf{x_i}))^2, & if \ y_{i} < f(\mathbf{x_i}) \\
&\tau (y_i - f(\mathbf{x_i})))^2, & if \ y_i \geq f(\mathbf{x_i}).
\end{aligned}
\right.
\end{equation}
The model tends to be overfitted with the increasing number of covariates.  To address the overfitting issue, a $L_2$ penalty is added to the risk function,
\begin{equation}\label{loss}
\mathcal{R}(\tau)=\frac{1}{n}\sum_{i=1}^{n} L_{\tau}(y_i,f(\mathbf{x_i}))+ \lambda\sum_{p=1}^{P} \sum_{q=1}^{Q}\left((w_{pq}^{(1)})^2 + (w_{q}^{(2)})^2 \right)^2.
\end{equation}

The loss function for ENN is differentiable everywhere, and therefore we can obtain the estimator of ENN by using gradient-based optimization algorithms (e.g., quasi-Newton Broyden–Fletcher–Goldfarb–Shanno (BFGS) optimization algorithm).

\subsection{Theoretical result}
In ENN, $\tau-$expectiles $f_{L_{\tau},P}^{\ast}$  can be estimated by minimizing the asymmetric least squares (ALS) loss,
\begin{equation}
\mathcal{R}_{L_{\tau},P}(f_{L_{\tau},P}^{\ast})=\mathcal{R}_{L_{\tau},P}^{\ast}=inf \lbrace \mathcal{R}_{L_{\tau},P}(f)=\int_{X\times Y} L_{\tau}(y,f(x))dP(x,y)|f: X \rightarrow \mathbb{R}   \text{ measurable} \rbrace,
\end{equation}
where $P$ is the distribution on $X\times Y$ and $f: X \rightarrow \mathbb{R}$ is some predictor.

$f_{L_{\tau},P}^{\ast}$ can be estimated by minimizing the risk function. Intuitively, the convergence rate of $f_{L_{\tau},P}^{\ast}$ is related to risk function $\mathcal{R}_{L_{\tau},P}(f).$ The following theorem describe the upper bound and lower bound of error of $f_{L_{\tau},P}^{\ast}$.
\begin{thm}
Let $L_{\tau}$ be the ALS loss function and $P$ be the distribution on $X\times Y$. We further assume that $f_{L_{\tau},P}^{\ast} < \infty$ is the conditional $\tau-$expectile for fixed $\tau \in (0,1)$. Then, for an arbitrary neural network function $f$, we have
$$
C_{\tau}^{-1/2}(\mathcal{R}_{L_{\tau},P}(f)-\mathcal{R}_{L_{\tau},P}^{*})^{1/2} \leq || f-f_{L_{\tau},P}^{*} ||_{L_{2}(P_{\mathbf{x}})} \leq c_{\tau}^{-1/2}(\mathcal{R}_{L_{\tau},P}(f)-\mathcal{R}_{L_{\tau},P}^{*})^{1/2},
$$
where $c_{\tau}=min \lbrace \tau, 1-\tau \rbrace$, $C_{\tau}=max \lbrace \tau, 1-\tau \rbrace.$
\end{thm}
Proof of this theorem can be found in the appendix of the paper.

\section{Simulation}
Simulation studies were conducted to compare the performance of ENN and ER under different settings. The genetic data used in the simulation is the real sequencing data from the $1000$ Genomes Project, located on Chromosome $17: 7344328 - 8344327$ \cite{DRA}. Totally $1000$ replicates were simulated for each simulation setting. In each replicate, we randomly selected a number of samples and SNPs from the $1000$ Genomes Project based on the simulation settings. Given the genotypes, we further simulated the phenotype by using different linear/non-linear functions or by assuming different types of interactions among SNPs or genes. \\

We divided the samples into training, validation, and testing sets. ENN and ER were applied to the training set to build models. While a variety of activation functions can be used in ENN, we choose ReLU due to its performance and computational advantage. Since the loss function of ENN is differentiable, we use the quasi-Newton BFGS optimization algorithm to estimate the parameters in ENN. We chose the starting point carefully to avoid the local minimum. To select a proper starting point, we generated a set of initial values from $U[-1,1]$, ran the algorithm for a few steps, and chose the initial values achieving the smallest loss as the initial values. Based on the initial values, the quasi-Newton BFGS optimization algorithm is implemented to iteratively estimate the parameters until the convergence criterion is satisfied. The models built on the training set were then applied to the validation set to choose the most parsimonious model with the optimal tuning parameter (i.e., $\lambda$). To choose the best $\lambda$, we use the grid search with values of 0,0.1,1,10,100.  This final model was then evaluated on the testing set by using the mean squared error (MSE).

\subsection{Simulation I - nonlinear relationship} 
In simulation I, we varied the relationships between genotypes and phenotypes, and compared the performances of ENN and ER. Specially, we considered the following four nonlinear functions to simulate the relationship between genotypes and phenotypes. For comparison purpose, we also include a linear function,
\begin{enumerate}
\item linear function: $y = \alpha + \epsilon, \alpha = \mathbf{x}^T\bm{\beta}$,
\item Hyperbolic function: $y = \frac{|\alpha|}{(1+ |\alpha|)} + \epsilon,\alpha = \mathbf{x}^T\bm{\beta}$,
\item Mixed function: $y = sin(\alpha) + 2*exp(-16\alpha^2) + \epsilon, \alpha = \mathbf{x}^T\bm{\beta}$,
\item Quadratic function: $y = \alpha^2 + \epsilon, \alpha = \mathbf{x}^T\bm{\beta}$,
\item Cubic function: $y = \alpha^3 + \epsilon, \alpha = \mathbf{x}^T\bm{\beta} $,
\end{enumerate}
where $\mathbf{x}$ is the vector of SNPs (coded as 0, 1 or 2), $\bm{\beta}$ represents the genetic effects generated from the uniform distribution of $U(-1,1)$, and $\epsilon \sim N(0,1)$. Totally $1000$ replicates were simulated. For each replicate, We randomly choose 500 samples and 50 SNPs from the 1000 Genomes Project.
\begin{figure}[h]
\begin{center}
\caption{Performance comparison between ENN and ER under various relationships between genotypes and phenotypes and different expectiles (i.e., 0.1, 0.25, 0.5, 0.75, and 0.9)}
\includegraphics[scale=0.4]{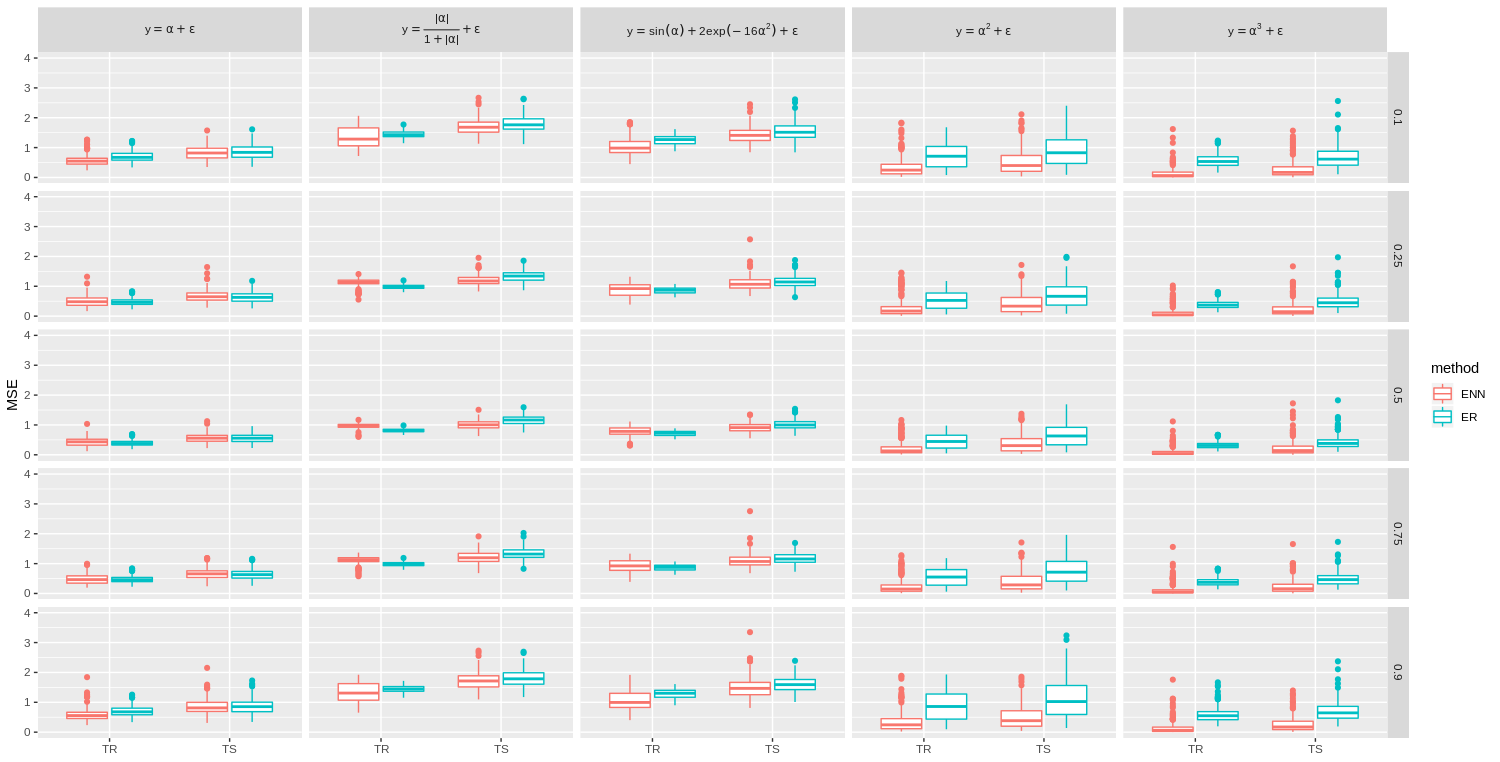}

\textbf{ENN:} expectile neural network; \textbf{ER:} expectile regression;
   \textbf{TR:} training; \textbf{TS:} testing; 
\end{center}
\end{figure}

The results from the simulation I are summarized in Figure 2. ENN outperforms ER in terms of MSE under four different nonlinear relationships, and has comparable performance with ER when the underlying relationship is linear. The pattern is consistent across different expectiles (i.e., 0.1, 0.25, 0.5, 0.75, and 0.9). While ENN outperforms ER for all four non-linear cases, ENN attains its best performance relative to ER when the underlying relationship is a high-order polynomial function (i.e., a cubic function).

\subsection{Simulation II - interactions among SNPs}
In simulation II, we considered three types of interactions, including a two-way multiplicative interaction, a two-way threshold interaction, and a three-way interactions \cite{GWSD}. Similar to simulation I, we simulated $1000$ replicates for each type of interaction. For each replicate, 500 samples and 50 SNPs were chosen from the 1000 Genomes Project. Among the 50 SNPs, we randomly selected $20\%$ of SNPs and simulated different types of interactions among the selected SNPs. Based on the simulated data, we compared MSEs of ENN and ER. For the comparison purpose, we also included a baseline model without any interaction.

\begin{figure}[H]
\begin{center}
\caption{Performance comparison between ENN and ER for different types of interactions and different expectiles (i.e., 0.1, 0.25, 0.5, 0.75, and 0.9)}
\includegraphics[scale=0.4]{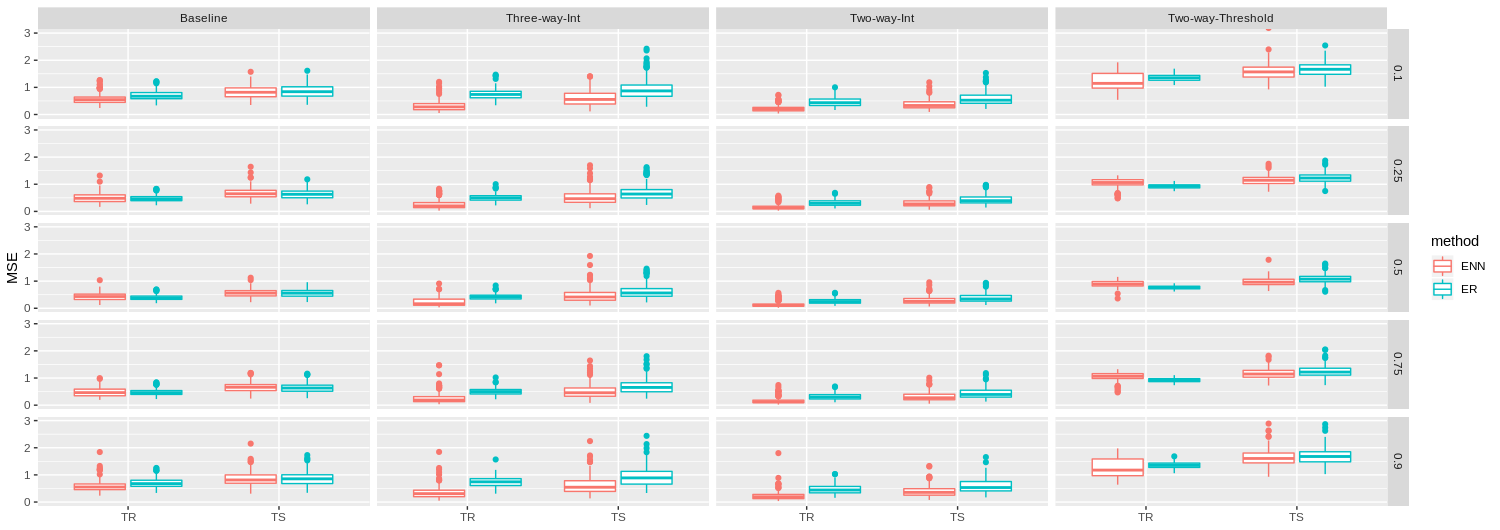}

 \textbf{ENN:} expectile neural network; \textbf{ER:} expectile regression;\textbf{TR:} training; \textbf{TS:} testing
\end{center}
\end{figure}

The results of the simulation II are summarized in Figure 3. Overall, ENN outperforms ER under all three interaction scenarios, indicating ENN's ability of taking interactions into account for improved performance. Among all interaction models, ENN attains its best performance relative to ER when there are three-way interactions. ENN also has more advantage over ER at the upper and lower expectiles (e.g., 0.1 and 0.9).  When there is no interaction, ENN has comparable performance with ER. 

\subsection{Simulation III - interactions between genes}
Investigating interactions among two or more than two genes is often interested in genetic studies. While a fully connected neural network can be built on all SNPs in the genes of interest, a neural network with a simpler architecture reflecting the underlying genetic data structure can be used to reduce the model's complexity and improve the model's performance. In this simulation, We illustrate the idea by modeling interactions between two genes with a non-fully connected architecture. In the non-fully connected architecture, the hidden units are only locally connected to SNPs in one gene (Figure 5). By using this simple architecture, we can reduce the number of parameters and build "gene-specific" hidden units to capture abstract features of a specific gene. To evaluate the performance of such an architecture, we simply simulated four SNPs for each gene, considered a two-way multiplicative interaction between two genes, and compared ENN with the non-fully connected architecture to ENN with a fully connected architecture.  
\begin{figure}[H]
\begin{center}
\caption{An alternative architecture for gene-gene interaction analyses}
\includegraphics[scale=0.6]{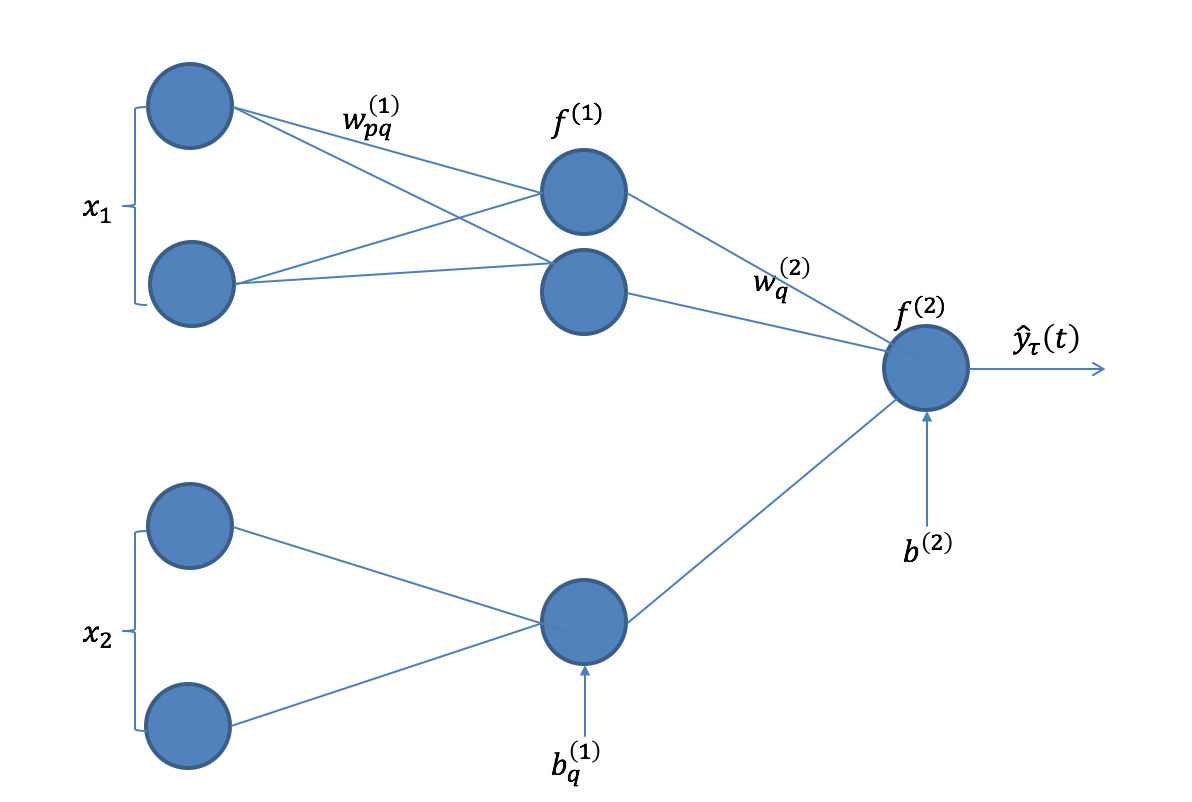}
\end{center}
\end{figure}

\begin{figure}[h]
\begin{center}
\caption{Performance comparison between ENN with a fully connected architecture and ENN with a non-fully connected architecture for gene-gene interaction analyses}
\includegraphics[scale=0.4]{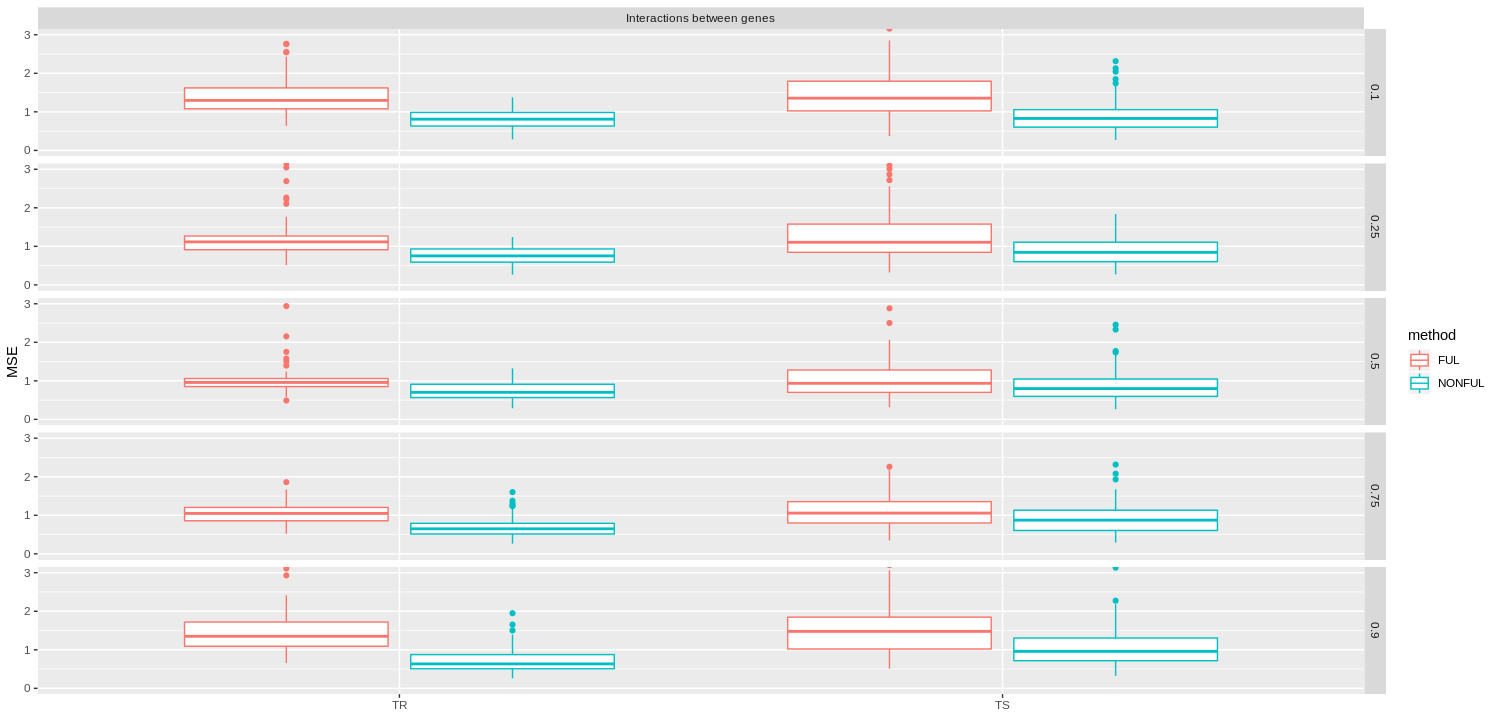}

\textbf{ENN:} expectile neural network; \textbf{TR:} training; \textbf{TS:} testing
\end{center}
\end{figure}
Figure 5 summarizes the results from simulation III. The results show that ENN with the Non-fully connected architecture attains lower MSE than ENN with the fully-connected architecture. As expected, the non-fully connected architecture requires fewer parameters and more reflects the underlying genetic data structure (i.e., genes are separate functional units), and therefore attains better performance than the fully-connected architecture. 

\section{Real data applications}
\subsection{The relationship between candidate SNPs with smoking quantities} 
We applied both ENN and ER to the genetic data from the Study of Addiction: Genetics and Environment(SAGE). The participants of the SAGE are selected from three large and complementary studies: the Family Study of Cocaine Dependence(FSCD), the Collaborative Study on the Genetics of Alcoholism(COGA), and the Collaborative Genetic Study of Nicotine Dependence(COGEND). In this application, we selected 155 SNPs, which were previously shown to have a potential role in nicotine dependence. After quality control, 149 SNPs remained for the analysis.  There are a total of 3897 samples in the SAGE data from different ethnic groups. We only included 3888 Caucasian and African American samples due to the small sample size of other ethnic groups.  Our interest is to use ENN and ER to build models on 149 SNPs, 3 covariates (i.e., sex, age, and race), and smoking quantities, which is measured by the largest number of cigarettes smoked in 24 hours. We divided the whole sample into the training, validation and test samples in the ratio of 3:1:1 to build the models, select the turning parameter, and evaluate the models, respectively.

\begin{table}[h]
\caption{The accuracy performance of two models built by ENN and ER based on 149 candidate SNPs and 3 covariates} \label{comparsion with two models}
\begin{center}
\begin{tabular}{|l|l|l|l|l|}
\hline
       & \multicolumn{2}{c|}{ENN} & \multicolumn{2}{c|}{ER} \\ \hline
$\tau$ & Train        & Test      & Train       & Test      \\ \hline
0.1    & 409.612       &  678.331         &  504.215           & 694.809          \\ \hline
0.25   & 346.118       &  579.164         &  394.836           & 588.759          \\ \hline
0.5    & 358.783       &  502.752         &  342.144           & 535.925          \\ \hline
0.75   & 344.399       &  604.969         &  421.955          &  613.676         \\ \hline
0.9    & 570.994       &  809.733         &  699.654          &  882.781         \\ \hline
\end{tabular}
\end{center}
\end{table}

Table 1 summarizes MSE of the models built by ENN and ER for five expectile levels (i.e., $\tau$= 0.1, 0.25, 0.5, 0.75, and 0.9). Table 1 shows that ENN outperforms ER, indicating the possibility of non-linear or non-additive effects among candidate SNPs and covariates. To provide a comprehensive view of the conditional distribution of smoking quantity, we ordered the expectiles estimated from ENN from lowest to highest and plotted their values for all five expectile levels. Figure 6 shows that the distributions of estimated expectiles are different across five expectile levels. When $\tau$= 0.5, ENN models the mean response, in which the estimated expectiles are similar for all individuals. Nonetheless, for high expectile levels (e.g., $\tau$= 0.9), the estimated expectiles vary among individuals and high-ranked individuals have much higher expectiles than low-ranked individuals.    

\begin{figure}[h]
\begin{center}
\caption{A comprehesive view of the conditional distribution of smoking quantity for five expectile levels (i.e., 0.1, 0.25, 0.5, 0.75, and 0.9)}
\includegraphics[scale=0.4]{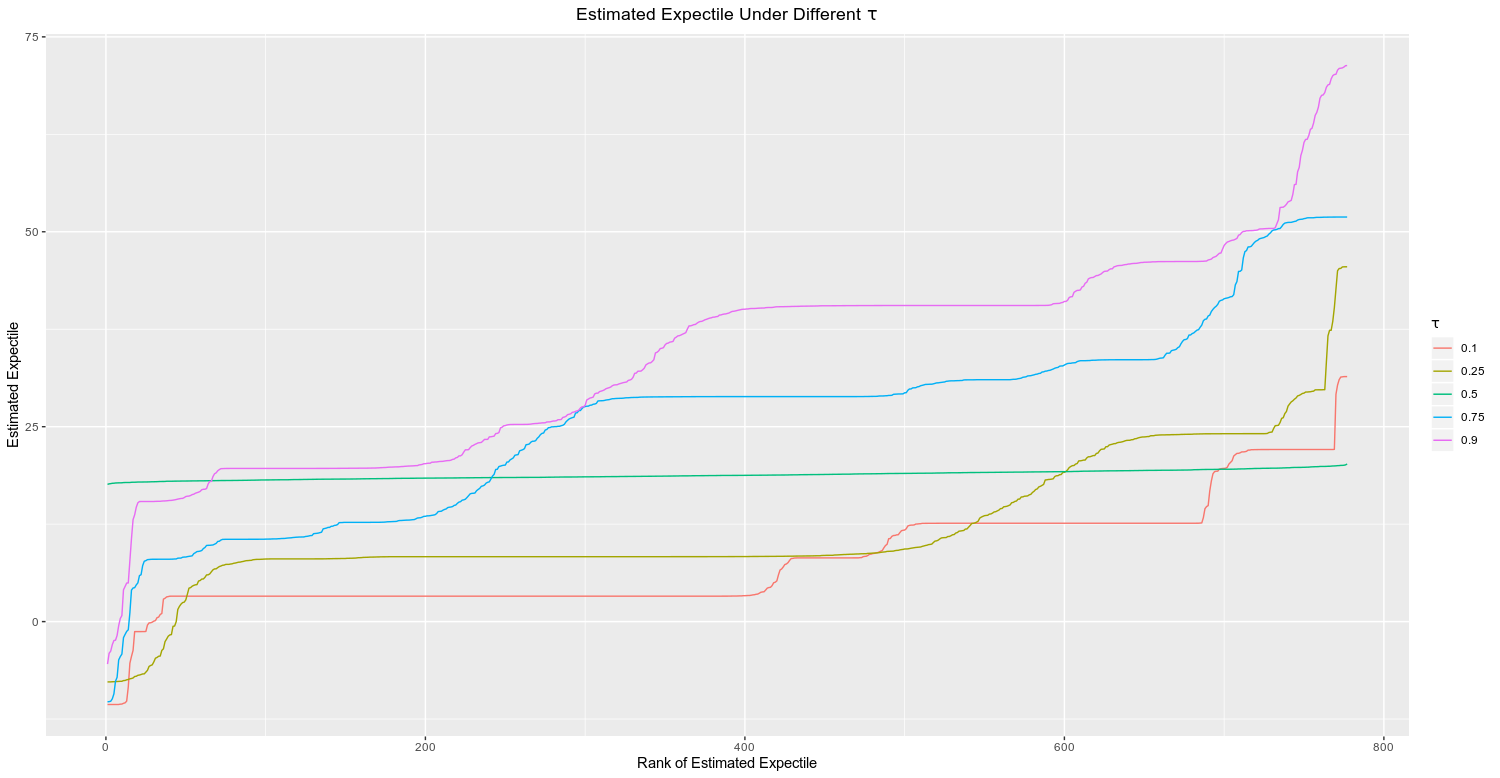}
\end{center}
\end{figure}

\subsection{Gene-gene interactions between the the $CHRNA5$–$CHRNA3$–$CHRNB4$ gene cluster} 
Previous evidence suggested potential interactions between the neuronal nicotinic acetylcholine receptors (nAChRs) subunit genes. In the second data analysis, we focused on the $CHRNA5$–$CHRNA3$–$CHRNB4$ gene cluster, and evaluated potential interactions by using ENN and ER. We consider three pairwise interactions between $CHRNA5$ and $CHRNA3$, $CHRNA5$ and $CHRNB4$, $CHRNA3$ and $CHRNB4$. The phenotype of interest in this analysis is the number of cigarettes smoked per day (CPD), which has been popularly used in the genetic study of nicotine dependence.

\begin{table}[h]
\caption{Evaluating a pairwise interaction between $CHRNA5$ and $CHRNA3$ by using ENN and ER} \label{comparsion with two models}
\begin{center}
\begin{tabular}{|l|l|l|l|l|}
\hline
       & \multicolumn{2}{c|}{ENN} & \multicolumn{2}{c|}{ER} \\ \hline
$\tau$ & Train        & Test      & Train       & Test      \\ \hline
0.1    & 1.106      &  2.022         &  1.183           & 2.036          \\ \hline
0.25   & 0.994       &  1.699         &  1.027           & 1.737          \\ \hline
0.5    & 0.896       &  1.266        &  0.908           & 1.304          \\ \hline
0.75   & 1.148       &  1.045         &  1.136          &  1.066         \\ \hline
0.9    & 2.015       &  1.335         &  2.069          &  1.357         \\ \hline
\end{tabular}
\end{center}
\end{table}

\begin{table}[h]
\caption{Evaluating a pairwise interaction between $CHRNA5$ and $CHRNB4$ by using ENN and ER} \label{comparsion with two models}
\begin{center}
\begin{tabular}{|l|l|l|l|l|}
\hline
       & \multicolumn{2}{c|}{ENN} & \multicolumn{2}{c|}{ER} \\ \hline
$\tau$ & Train        & Test      & Train       & Test      \\ \hline
0.1    & 1.139      &  2.020         &  1.186           & 2.049          \\ \hline
0.25   & 0.980       &  1.701         &  1.029           & 1.735          \\ \hline
0.5    & 0.901       &  1.277        &  0.908           & 1.305          \\ \hline
0.75   & 1.149       &  1.047         &  1.136          &  1.071         \\ \hline
0.9    & 2.054       &  1.318         &  2.070          &  1.351         \\ \hline
\end{tabular}
\end{center}
\end{table}

\begin{table}[h]
\caption{Evaluating a pairwise interaction between $CHRNA3$ and $CHRNB4$  by using ENN and ER} \label{comparsion with two models}
\begin{center}
\begin{tabular}{|l|l|l|l|l|}
\hline
       & \multicolumn{2}{c|}{ENN} & \multicolumn{2}{c|}{ER} \\ \hline
$\tau$ & Train        & Test      & Train       & Test      \\ \hline
0.1    & 1.133      &  2.019         &  1.183           & 2.035          \\ \hline
0.25   & 0.979       &  1.683         &  1.020           & 1.696          \\ \hline
0.5    & 0.892       &  1.278        &  0.896           & 1.279         \\ \hline
0.75   & 1.150       &  1.048         &  1.128          &  1.081         \\ \hline
0.9    & 2.020       &  1.342         &  2.040          &  1.386         \\ \hline
\end{tabular}
\end{center}
\end{table}

Tables 2-4 summarize MSE of the interaction models built by using ENN and ER for five expectile levels. For all 3 scenarios, expectile neural network outperforms expectile regression in terms of MSE. To graphically view the conditional distribution of CPD, we ranked the expectiles estimated from ENN and plotted the values against the estimated expectiles. Overall, the estimated expectiles tends to be similar when $\tau$= 0.5 (i.e., mean), while they are quite different for high expectile levels (e.g., $\tau$= 0.9). This suggest that the gene-gene interactions may play a more important role in models with high expectiles than the mean models.    

\begin{figure}
\begin{center}
\includegraphics[scale=0.39]{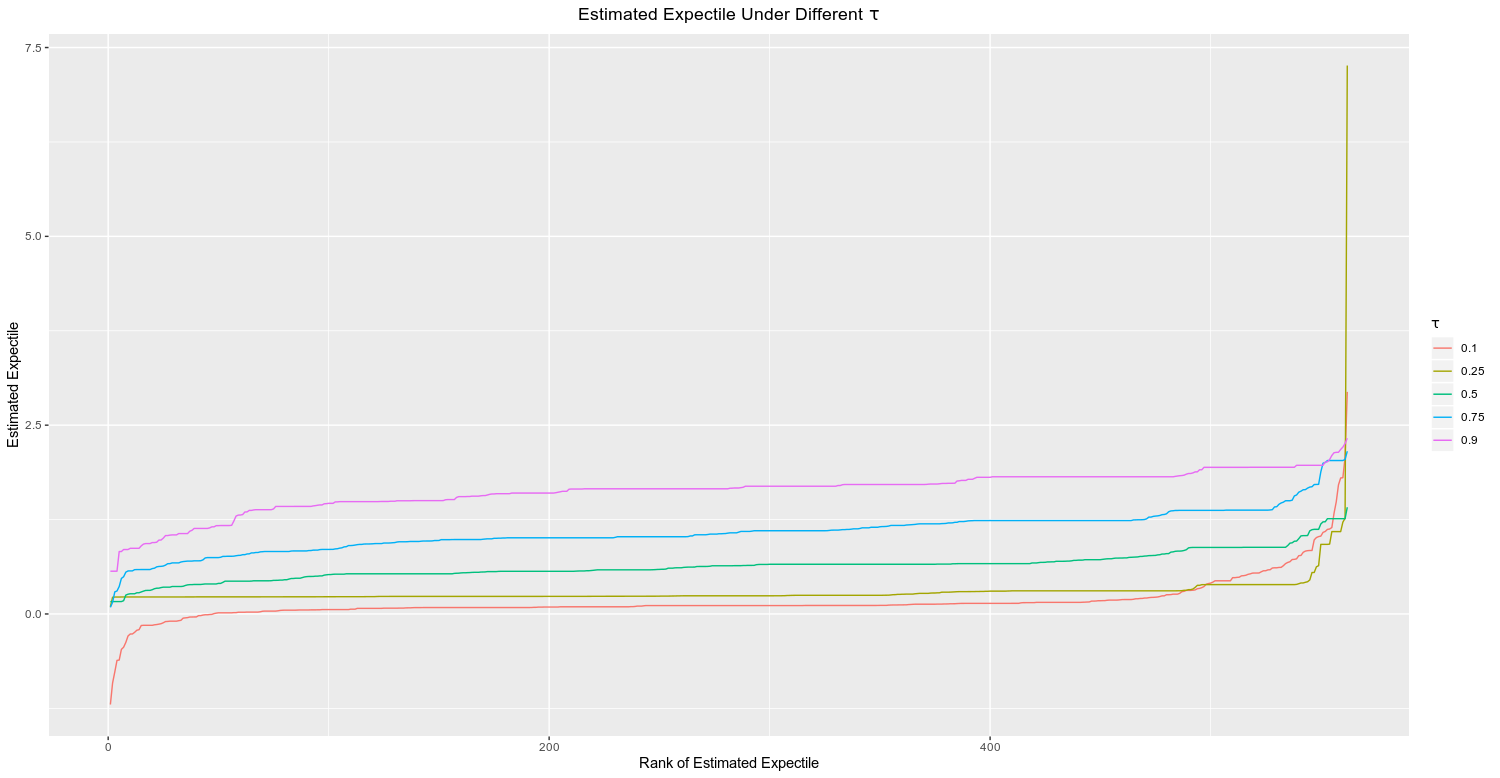}
\caption{The conditional distribution of CPD considering the interaction between $CHRNA5$ and $CHRNA3$}
\end{center}
\end{figure}

\begin{figure}
\begin{center}
\includegraphics[scale=0.39]{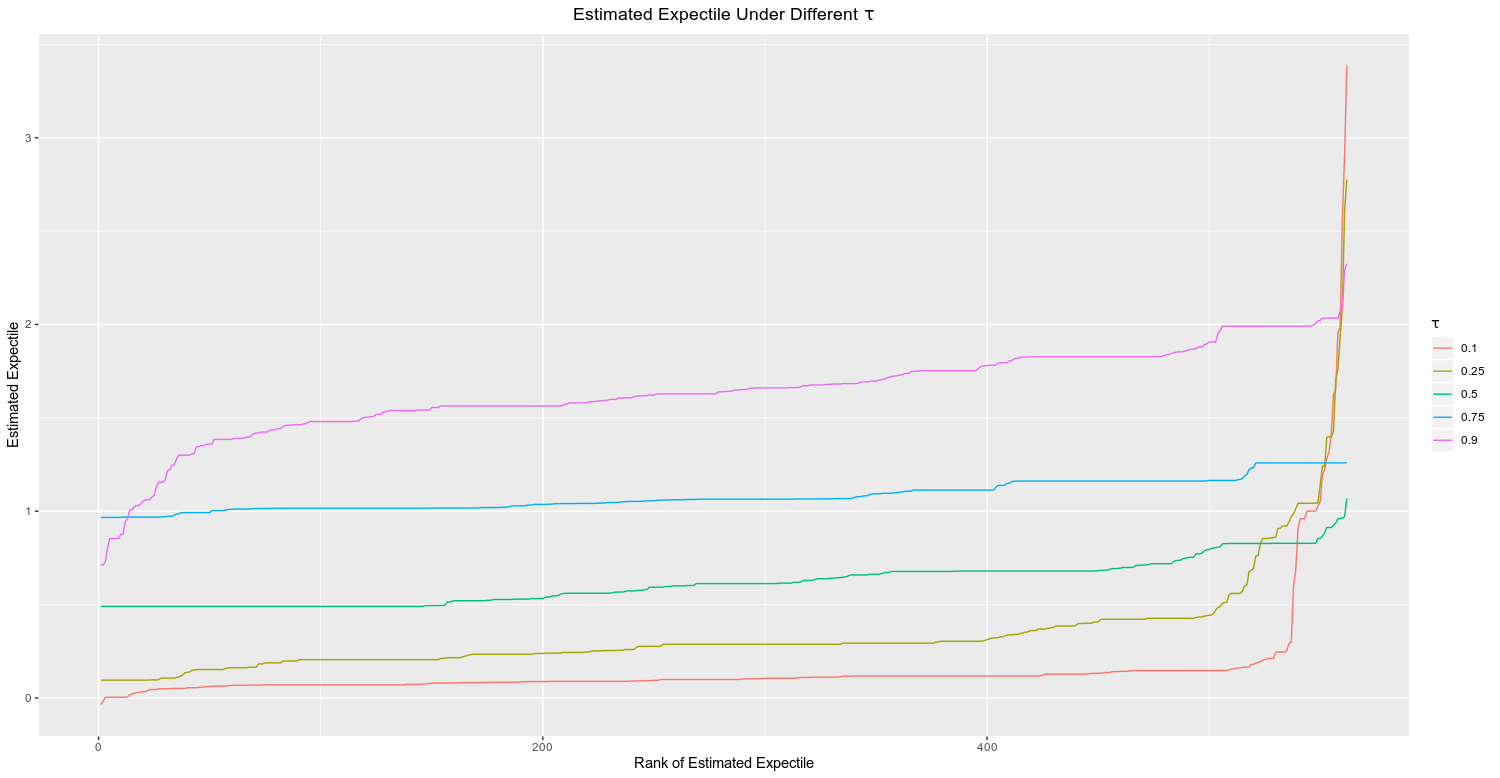}
\caption{The conditional distribution of CPD considering the interaction between $CHRNA5$ and $CHRNB4$}
\end{center}
\end{figure}

\begin{figure}
\begin{center}
\includegraphics[scale=0.39]{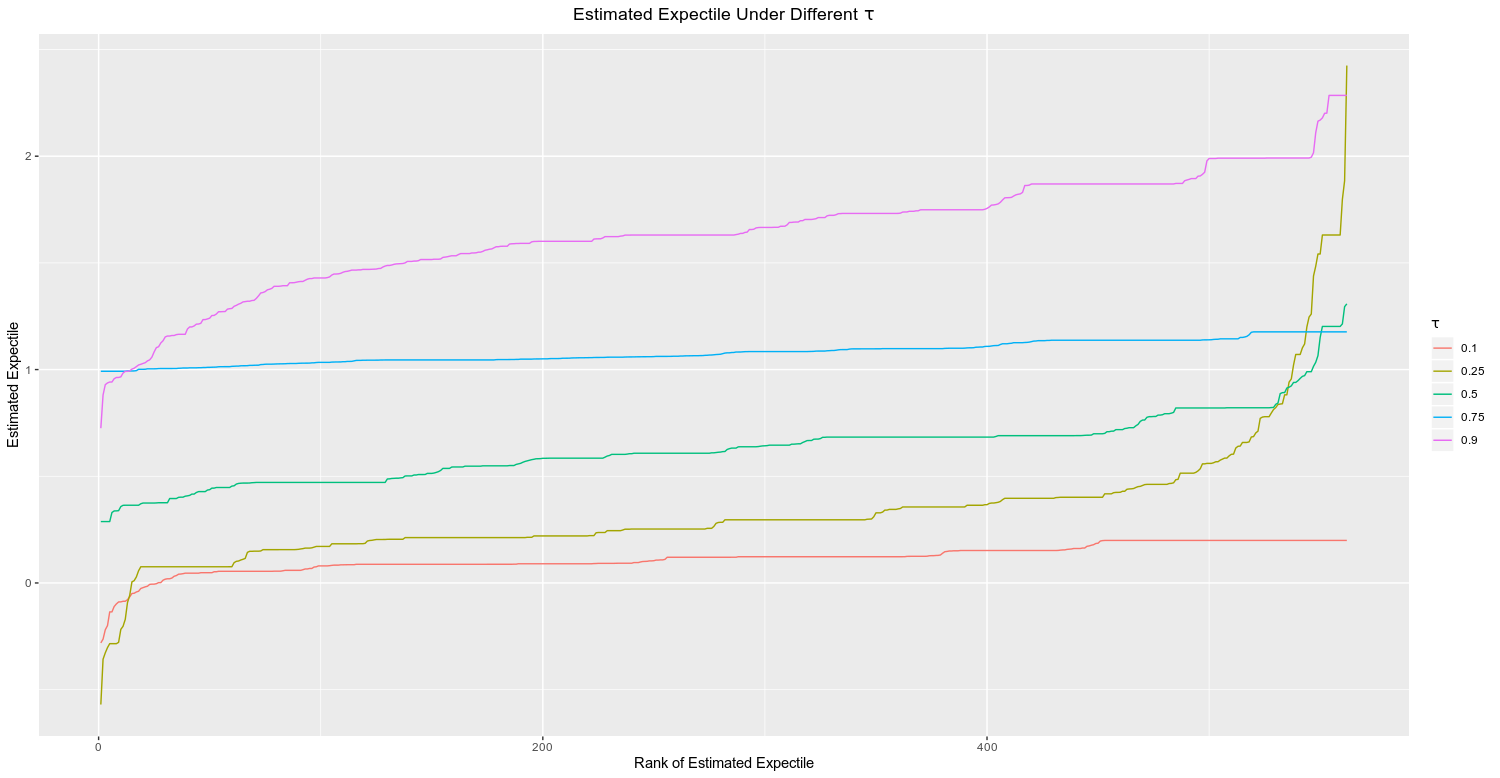}
\caption{The conditional distribution of CPD considering the interaction between $CHRNB4$ and $CHRNA3$}
\end{center}
\end{figure}

\section{Summary and discussion}
In this paper, we develop an ENN method, which inherits advantages from both neural networks and expectile regression. Using the hierarchical structure from neural networks, ENN can learn complex and abstract features from genotypes, making it suitable for modeling the complex relationship between genotypes and phenotype. Similar to ER, ENN can also explore the conditional distribution and provide a comprehensive view of the genotype-phenotype relationship.

Through simulations and a real data application, we demonstrate that ENN outperforms ER when there are non-additive and non-linear effects. Evidence also suggests that ENN has more advantages than ER when the model involves high-order interaction effects or non-linear effects. This may suggest ENN has improved performance when the underlying genotype-phenotype relationships become more complicated. The real data analysis shows that genetic effects can vary among different expertiles. Compared to the classical linear regression, ENN provides us more information about the genotype-phenotype relationship via the conditional distributions for different expectile levels. 

While regularization has been incorporated into ENN to avoid overfitting, ENN can still be subject to overfitting when the number of SNPs become extremely large (e.g., one million). To deal with such a large number of SNPs, we can model the overall genetic effect as a random effect and extend ENN, which is an interesting topic for future work. In this paper, we focus on introducing the ENN model and providing an inequality that bounds the integrated squared error of an expectile function estimator. Statistical properties of ENN (e.g., rate of convergence) are also important topics that worth further investigation in the future.

\section*{Acknowledgment}
This work was supported by NIH 1R01DA043501-01 and NIH 1R01LM012848-01. Funding support for the Study of Addiction: Genetics and 
Environment was provided through the NIH Genes, Environment and Health Initiative [GEI] (U01 HG004422). The SAGE datasets used for the 
analyses were obtained from dbGaP at $https://www.ncbi.nlm.nih.gov/projects/gap/cgi-bin/study.cgi?study_id=phs000092.v1.p1$ through dbGaP 
accession number phs000092.v1.p1.

\section{Appendix}
\begin{deft}
Let $L_{\tau}: Y\times \mathbb{R} \rightarrow [0,\infty)$ be the asymmetric least square loss function and $Q$ be a distribution on $Y=[-M, M]$. Then, the inner $L_{\tau}-risks$ of $Q$ could be defined as
$$
\mathcal{C}_{\tau,Q}(t)= \int_{Y} L_{\tau}(y,t)dQ(y), t = f(\mathbf{x_i}) \in \mathbb{R},
$$
and the minimal inner $L_{\tau}-risk$ is
$$
\mathcal{C}^{*}_{L_{\tau},Q}= inf_{t \in \mathcal{R}} \mathcal{C}_{L_{\tau},Q}(t)
$$
\end{deft}

\begin{lem}
Let $L_{\tau}$ be the asymmetric least square loss function and $Q$ be a distribution on $\mathbb{R}$ with $\mathcal{C}_{L_{\tau},Q}^{*}<\infty$. For a fixed $\tau \in (0,1)$ and for all $t \in \mathbb{R}$, we have
\begin{align*}
c_{\tau}(t-t^{*})^2 \leq \mathcal{C}_{L_{\tau},Q}(t)-\mathcal{C}^{*}_{L_{\tau},Q} \leq C_{\tau}(t-t^{*})^{2},
\end{align*}
where $c_{\tau}=min \lbrace \tau, 1-\tau \rbrace$ and $C_{\tau}= max \lbrace \tau, 1-\tau \rbrace$, $t^{*}$ is $\tau-$expectile .
\end{lem}

\begin{proof}
Let us fix $\tau \in (0,1)$. We use the result obtained in Newey and Powell $\cite{ALSE}$. For a distribution $Q$ on $\mathcal{R}$ satisfies $\mathcal{C}_{L_{\tau},Q}^{*} < \infty$, the $\tau-$expectile $t^{*}$ is the only solution of
\begin{align*}
\tau \int_{y \geq t^{*}} (y-t^{*}) dQ(y)= (1-\tau)\int_{y < t^{*}} (t^{*}-y) dQ(y). \numberthis
\end{align*}
First, We consider the lower bound.\\
To obtain the inner $L_{\tau}-$risks of Q, we consider two cases: $t \geq t^{*}$ and $t < t^{*}$.\\
When $t \geq t^{*}$, we have
\begin{align*}
\int_{y<t}(y-t)^{2}dQ(y) & = \int_{y<t}(y-t^{*}+t^{*}-t)^{2}dQ(y) \\
&=\int_{y<t}(y-t^{*})^{2}dQ(y)+2(t^{*}-t)\int_{y<t}(y-t^{*})dQ(y)+(t^{*}-t)^{2}Q((-\infty,t)) \\
&=\int_{y<t^*}(y-t^{*})^{2}dQ(y)+\int_{t^{*} \leq y <t}(y-t^{*})^{2}dQ(y )+(t^{*}-t)^{2}Q((-\infty,t))\\
& + 2(t^{*}-t) \int_{y<t^{*}}(y-t^{*})dQ(y)+2(t^{*}-t)\int_{t^{*} \leq y <t}(y-t^{*})dQ(y),
\end{align*}
and
\begin{align*}
\int_{y \geq t}(y-t)^{2}dQ(y)&=\int_{y\geq t^{*}}(y-t^{*})^{2}dQ-\int_{t^{*} \leq y < t}(y-t^{*})^{2}dQ(y)+(t^{*}-t)^{2}Q([t,\infty))\\
& 2(t^{*}-t)\int_{y\geq t^{*}}(y-t^{*})dQ(y)-2(t^{*}-t)\int_{t^{*}\leq y <t}(y-t^{*})dQ(y).
\end{align*}
By definition and (13), we have
\begin{align*}
\mathcal{C}_{L_{\tau},Q}(t) & = (1-\tau)\int_{y<t}(y-t)^{2}dQ(y)+\tau\int_{y \geq t}(y-t)^{2}dQ(y)\\
& = (1-\tau) \int_{y<t^{*}}(y-t^{*})^{2}dQ(y)+\tau \int_{y \geq t^{*}}(y-t^{*})dQ(y)\\
& + 2(t^{*}-t)((1-\tau) \int_{y<t^{*}}(y-t^{*})dQ(y)+\tau\int_{y \geq t^{*}} (y-t^{*})dQ(y))\\
& + (t^{*}-t)^{2}(1-\tau)Q((-\infty,t))+(t^{*}-t)^{2} \tau Q([t,\infty))\\
& + (1-2\tau)\int_{t^{*} \leq y <t}(y-t^{*})^{2}dQ(y)+2(1-2\tau)(t^*-t)\int_{t^{*} \leq y<t}(y-t^{*})dQ(y)\\
& = \mathcal{C}_{L_{\tau},Q}(t^{*}) + (t^{*}-t)^2(1-\tau)Q((-\infty,t))+(t^{*}-t)^{2}\tau Q([t,\infty))\\
& +(1-2\tau)\int_{t^{*} \leq y < t}(y-t^{*})^{2} + 2(t^{*}-t)(y-t^{*})dQ(y)\\
\end{align*}
Therefore,
\begin{align*}
&\mathcal{C}_{L_{\tau},Q}(t)-\mathcal{C}_{L_{\tau},Q}(t^{*})\\
& = (t^{*}-t)^{2}(1-\tau)Q((-\infty,t^{*})) + (t^{*}-t)^{2}(1-\tau)Q([t^{*}, t))+ (t^{*}-t)^{2}\tau Q([t,\infty))\\
& +(1-2\tau)\int_{t^{*} \leq y <t} (y-t^{*})^{2} +2(t^{*}-t)(y-t^{*})dQ(y)\\
& = (t^{*}-t)^{2}((1-\tau)Q((-\infty,t^{*}))+\tau Q([t, \infty))) - \tau \int_{t^{*} \leq y < t}(y-t^{*})^{2} + 2(t^{*}-t)(y-t^{*})dQ(y)\\
& + (t^{*}-t)^{2}(1-\tau)Q([t^{*},t)) + (1-\tau)\int_{t^{*} \leq t <t }(y-t^{*})^{2} + 2(t^{*}-t)(y-t^{*})dQ(y)\\
& = (t^{*}-t)^{2}((1-\tau)Q((-\infty,t^{*}))+ \tau Q([t,\infty)))-\tau\int_{t^{*} \leq y <t}(y-t^{*})(y+t^{*}-2t)dQ(y)\\
& + (1-\tau)\int_{t^{*} \leq y <t}(y-t^{*})^{2} + 2(t^{*}-t)(y-t^{*})+(t^{*}-t)^{2}dQ(y)\\
&= (t^{*}-t)^{2}((1-\tau)Q(-\infty,t^{*}))+\tau Q([t, \infty)))+ \tau \int_{t^{*} \leq y <t}(y-t^{*})(2t-t^{*}-y)dQ(y)\\
& (1-\tau)\int_{t^{*} \leq y <t}(y-t)^{2}dQ(y). \numberthis
\end{align*}
This leads to the lower bound of inner $L_{\tau}-$risk when $t \geq t^{*}$,
\begin{align*}
&\mathcal{C}_{L_{\tau},Q}(t)-\mathcal{C}_{L_{\tau},Q}(t^{*})\\
& \geq c_{\tau}(t^{*}-t)^{2}(Q((-\infty,t^{*}))+Q([t,\infty)))+c_{\tau} \int_{t^{*} \leq y \leq t}(y-t^{*})(2t-t^{*}-y)+(y-t)^2dQ(y)\\
& = c_{\tau}(t^{*}-t)^{2}(Q((-\infty,t^{*}))+Q([t,\infty)))+c_{\tau} \int_{t^{*} \leq y \leq t}(t^{*})^2-2tt^{*}+t^{2}dQ(y)\\
& = c_{\tau}(t^{*}-t)^{2}(Q((-\infty,t^{*}))+Q([t,\infty)))+c_{\tau}(t^{*}-t)^{2}Q([t^{*},t))\\
& = c_{\tau}(t^{*}-t)^{2}.
\end{align*}
When $t<t^{*}$, using similar arguments, we have
\begin{align*}
\mathcal{C}_{L_{\tau},Q}(t)-\mathcal{C}_{L_{\tau},Q}(t^{*}) & = (t^{*}-t)^{2}((1-\tau)Q((-\infty,t))+\tau Q([t^{*},\infty)))+\tau \int_{t \leq y <t^{*}}(y-t)^{2}dQ(y)\\
& + (1-\tau)\int_{t \leq y < t^{*}} (t^{*}-y)(y+t^{*}-2t)dQ(y)\\
& \geq c_{\tau}(t^{*}-t)^{2}.
\end{align*}
Therefore, we summarize them into one inequality
\begin{align*}
\mathcal{C}_{L_{\tau},Q}(t)-\mathcal{C}_{L_{\tau},Q}(t^{*}) \geq c_{\tau}(t^{*}-t)^{2}.
\end{align*}
Next, we consider the upper bound. Similarly, when $t \geq t^{*}$,
\begin{align*}
& \mathcal{C}_{L_{\tau},Q}(t)-\mathcal{C}_{L_{\tau},Q}(t^{*})\\
& \leq C_{\tau}(t^{*}-t)^{2}(Q((-\infty,t^{*}))+Q([t,\infty)))+C_{\tau}\int_{t^{*} \geq y <t }((y-t^{*})(2t-t^{*}-y)+(y-t)^{2})dQ(y)\\
& = C_{\tau}(t^{*}-t)^{2}. \numberthis
\end{align*}
For the case of $t < t^{*}$, the inequality still holds.
Combining these two inequality, we have
$$c_{\tau}(t-t^{*})^2 \leq \mathcal{C}_{L_{\tau},Q}(t)-\mathcal{C}^{*}_{L_{\tau},Q} \leq C_{\tau}(t-t^{*})^{2}.$$
\end{proof}

Based on the Lemma 1, we can prove Theorem 1.

\begin{proof}
If $x \in X$, we define $t = f(x)$ and $t^{*} = f_{L_{\tau},P}^{*}(x)$. By Lemma 1, for $Q=P(\cdot | x)$, we can get the following result
\begin{equation*}
C_{\tau}^{-1} \left(\mathcal{C}_{L_{\tau}, P(\cdot | x)}(f(x)))-\mathcal{C}^{*}_{L_{\tau},P(\cdot | x)}\right) \leq |f(x)-f^{*}_{L_{\tau},P}(x)|^2 \leq c_{\tau}^{-1}\left(\mathcal{C}_{L_{\tau}, P(\cdot | x)}(f(x)) - \mathcal{C}^{*}_{L_{\tau}, P(\cdot | x)}\right).
\end{equation*}
If we integrate it with respect to $P_{X}$ and take the square root, we can get the final result.
\end{proof}

\end{document}